\documentclass[letterpaper, 10pt, conference]{ieeeconf}
\IEEEoverridecommandlockouts

\usepackage[dvipsnames]{xcolor}

\usepackage{./TLC}

\usepackage{circuitikz}
\usepackage{blox}
\usepackage{pgfplots}
\pgfplotsset{compat=1.16}
\usepgfplotslibrary{groupplots}

\pgfplotscreateplotcyclelist{colors}{
        {CornflowerBlue},
        {BurntOrange},
        {RubineRed},
        {Emerald},
        {SpringGreen},
        {Goldenrod}
}

\usepackage{mathtools}
\usepackage{gensymb}
\usepackage{physics}
\usepackage{siunitx}

\usepackage{epigraph}
\usepackage{shellesc}
\IEEEtriggeratref{33}

\title{Monotone RLC circuits}

\author{Thomas Chaffey$^{1}$ \and Rodolphe Sepulchre$^{1}$%
\thanks{The research leading to these results has received funding from the European
Research Council under the Advanced ERC Grant Agreement Switchlet n. 670645.}%
\thanks{$^{1}$University of Cambridge, Department of Engineering, Trumpington Street,
Cambridge CB2 1PZ, {\tt\small tlc37@cam.ac.uk}, {\tt\small r.sepulchre@eng.cam.ac.uk}.}
}
\begin{document}
\maketitle

\begin{abstract}
        The circuit-theoretic origins of maximal monotonicity are revisited using
        modern optimization algorithms for maximal monotone operators.  We present an
        algorithm for computing the periodic behavior of an interconnection of maximal
        monotone systems using a fixed point iteration. The fixed point iteration 
        may be split according to the interconnection structure of the system.  In
        this preliminary work, the approach is demonstrated on port interconnections 
        of maximal monotone resistors and LTI capacitors and inductors.  
\end{abstract}

\section{Introduction}

        The property of \emph{maximal monotonicity} first arose in the
        study of nonlinear electrical circuits, in early efforts to extend the
        tractability of linear, time invariant, passive networks to networks
        containing nonlinear resistors.   The prototype of a maximal monotone
        element was Duffin's
        \emph{quasi-linear} resistor \cite{Duffin1946}, a nonlinear resistor with a
        non-decreasing $i-v$ characteristic.  Quasi-linearity was refined by Minty
        \cite{Minty1960, Minty1961, Minty1961a} to produce the modern concept of
        maximal monotonicity, in the context of an algorithm for solving networks of
        nonlinear resistors.  Desoer and Wu \cite{Desoer1974} studied existence and
        uniqueness of solutions to networks of nonlinear resistors, capacitors and
        inductors defined by maximal monotone relations.

        Following the influential paper of Rockafellar in 1976
        \cite{Rockafellar1976}, maximal monotonicity has grown to become a
        fundamental property in convex optimization \cite{Rockafellar1997, Ryu2016,
        Ryu2020, Parikh2013, Bertsekas2011, Combettes2011}, forming the basis of a large body of
        work on tractable first order methods for large scale and nonsmooth
        optimization problems, which have seen a surge of interest in the last
        decade.  However, the physical significance of maximal monotonicty in
        nonlinear circuit theory has been somewhat forgotten.

        In this paper, we revisit the study of nonlinear electrical networks using
        modern tools from convex optimization.  We study the problem of computing the
        periodic output of a system described by an interconnection of maximal
        monotone input/output relations, 
        which is forced by a periodic input.  The approach is demonstrated on the
        class of circuits formed by connecting maximal monotone resistors in series
        and parallel with LTI capacitors and inductors. A similar problem has been studied by Heemels
        \emph{et al.} \cite{Heemels2017}, for the class of systems described by a
        Lur'e-type feedback interconnection of a passive LTI state space system and a
        maximal monotone nonlinearity.  In this paper, we take a purely input/output
        view, and show that the output may be computed using a fixed point iteration
        in the space of periodic trajectories.  This computation may be split using a
        \emph{base splitting scheme} such as the Douglas-Rachford algorithm, and the
        splitting corresponds precisely to the interconnection structure.
        Computational steps are performed individually for each system component.
        This approach is reminiscent of frequency response analysis of LTI systems
        using the transfer function of their components.  Existing frequency response
        methods for nonlinear systems are either approximate and limited in their
        applicability, as in harmonic analysis
        \cite{Feldmann1996, Blagquiere1966, Krylov1947, Slotine1991}, or involve performing a
        transient simulation and waiting for convergence \cite{Aprille1972,
Brogliato2016, Cellier2006}. 

        The fundamental property which allows our splitting approach is the
        preservation of maximal monotonicity under port interconnections
        \cite{Camlibel2013}, or, in the language of circuit theory, series and 
        parallel interconnections.  This property is reminiscent of the fundamental
        theorem of passivity (preservation under parallel and negative feedback
        interconnections), and indeed, maximal monotonicity can be viewed as a
        generalization of passivity.  For a causal, LTI transfer function, passivity
        is precisely monotonicity on the signal space $L_2$ \cite{Desoer1975}.
        Maximal monotonicity provides a generalization of passivity theory to
        nonlinear systems which retains a connection to an algorithmic theory - in
        contrast, the connection between passive LTI systems and convex optimization 
        given by the KYP lemma \cite{Yakubovich1962, Kalman1963, Popov1964} 
        is lost for nonlinear systems.  Maximal monotonicity has been a central
        property in the study of nonsmooth dynamical systems \cite{Brogliato2020}, and early
        connections between passive systems and maximal monotone relations exist in
        this work - in particular, in the study of Lur'e systems with maximal
        monotone nonlinearities in the feedback path \cite{Camlibel2016,
        Brogliato2004}.  

        The remainder of this paper is organized as follows. In
        Section~\ref{sec:elements}, we introduce some basic theory of
        maximal monotonicity.  In section~\ref{sec:computation}, we 
        develop a computational technique to compute the periodic output of a periodically
        driven maximal monotone system. In section~\ref{sec:1ports}, we introduce the
        class of RLC circuits with maximal monotone, nonlinear resistors.
        Section~\ref{sec:example} gives a computational example on such a circuit. In
        Section~\ref{sec:cyclo}, we discuss the connection between monotonicity and
        cyclo-passivity.  We conclude in Section~\ref{sec:conclusions} with some open
        questions for future research.

\section{Maximal monotone relations}
\label{sec:elements}

\subsection{Relations}

We begin by introducing some mathematical preliminaries.

\begin{definition}
        A \emph{relation} on a space $X$ is a subset $S \subseteq X \times X$.
\end{definition}

%For example, a resistor might be represented by a relation on $\R$:
%\begin{IEEEeqnarray*}{rCl}
%        S_{\text{resistor}} &=& \left\{ (i, v) \in \R \times \R \; \middle| \; v = Ri \right\}.
%\end{IEEEeqnarray*}
%In this case, the relation describes a linear function.  In general, however, a
%relation may be set-valued.  
We write $y \in S(u)$ to denote $(u, y) \in S$.

The usual operations on functions can be extended to relations:
\begin{IEEEeqnarray*}{rCl}
        S^{-1} &=& \{ (y, u) \; | \; y \in S(u) \}\\
        S + R &=& \{ (x, y + z) \; | \; (x, y) \in S,\; (x, z) \in R \}\\
        SR &=& \{ (x, z) \; | \; \exists\, y \text{ s.t. } (x, y) \in R,\; (y, z) \in S \}.
\end{IEEEeqnarray*}
Note that the relational inverse $S^{-1}$ always exists, but in general, $S S^{-1}
\neq I$, where $I$ is the identity relation $\{(x, x)\;|\;x \in X\}$.

\subsection{Maximal Monotonicity}

The property of \emph{monotonicity} connects the physical property of energy
dissipation in a device to algorithmic analysis methods.  
Let $\mathcal{H}$ be a
Hilbert space with inner product $\bra{\cdot}\ket{\cdot}$ and induced norm
$\norm{x} = \sqrt{\bra{x}\ket{x}}$.  

Monotonicity on $\mathcal{H}$ is defined as follows.

\begin{definition}
        A relation $S \subseteq \mathcal{H}\times\mathcal{H}$ is called \emph{monotone} if
\begin{IEEEeqnarray*}{rCl}
        \langle u_1 - u_2 | y_1 - y_2 \rangle \geq 0
\end{IEEEeqnarray*}
for any $(u_1, y_1), (u_2, y_2) \in S$.  
A monotone relation is called \emph{maximal} if it is not properly contained in any
other monotone relation.
\end{definition}

Note that this definition refers to monotonicity in the operator theoretic sense, and
this is distinct from the notion of monotonicity in the sense of partial order
preservation by a state-space system (see, for example, \cite{Angeli2003}).

Monotonicity is preserved under a number of operations. 
The proof of the following lemma may be found in \cite{Ryu2020}.

\begin{lemma}
        \label{lem:monotone_properties}
        Consider relations $G$ and $F$ which are monotone on $\mathcal{H}$.  Then
        \begin{enumerate}
                \item $G^{-1}$ is monotone; \label{inversion}
                \item $G + F$ is monotone; \label{sum}
                \item $\alpha G$ is monotone for $\alpha > 0$;
                \item $C(x, y) = \left\{(u, v)\; \middle|\; u \in G(x), v \in F(y)
                        \right\}$ is monotone; \label{concatenation}
                \item If $F$ is monotone on a space of dimension $s$, and $M \in
                        \R^{s \times t}$, then the relation given by
                        \begin{equation*}
                                G(x) = M\tran F (Mx)
                        \end{equation*}
                        is also monotone. \label{congruence}
        \end{enumerate}
\end{lemma}

Maximality is preserved under inversion.  However, in general, maximality is not 
preserved when two relations are added (indeed, 
their sum may be empty). We make the following assumption on summations throughout
the rest of this paper, which guarantees maximality of the sum \cite[Thm. 1]{Rockafellar1970}.

\begin{assumption}
        \label{ass:domains}
        Any summation of two relations $G$ and $F$ obeys
                        \begin{IEEEeqnarray*}{rrCl}
                               & \interior \dom F \cap \dom G &\neq& \varnothing\\
                                \text{or } & \interior \dom G \cap \dom F &\neq&
                                \varnothing,
                        \end{IEEEeqnarray*}
                        where $\dom S$ denotes the domain of the relation $S$.
\end{assumption}

This assumption is sufficient (but not necessary) for the existence of solutions to
the summation (that is, the resulting relation is nonempty).  We omit the proof of this fact.

The preservation of monotonicity under the properties of
Lemma~\ref{lem:monotone_properties} allows us to show that several important system types are monotone
by testing the monotonicity of their component subsystems.  This idea is
explored by \c{C}amlibel and van der Schaft \cite{Camlibel2013}.  Here, we observe
that monotonicity is preserved under series and parallel port interconnections of
electrical elements.  This result forms the basis for the algorithmic methods to
follow.

Electrical 1-port circuits have two external terminals.  
Two variables may be measured across these
terminals - the port voltage $v$ and the magnetic flux linkage $\phi$.  Two variables may be measured 
through these terminals - the port current $i$ and the charge $q$.  We assume each of
these variables takes values in $\R$.  A 1-port is defined
by a relation between $v$ and $i$.
1-port circuits may be connected in series and parallel to form a new 1-port, with the interconnected
system constrained by Kirchoff's current and voltage laws.  These interconnections
preserve monotonicity.

\begin{proposition}\label{prop:series_parallel}
        Given two electrical 1-port elements $R_1$ and $R_2$, each represented by
        relations on a scalar Hilbert space $\mathcal{H}$ mapping current $i$ to voltage $v$, their series 
        interconnection
        \begin{enumerate}
        \item KCL: $i = i_1 = i_2$
        \item Device: $(i_1, v_1) \in R_1$, $\qquad(i_2, v_2) \in R_2$
        \item KVL: $v_1 + v_2 = v$,
        \end{enumerate}
        is a monotone relation on $\mathcal{H}$ from $i$ to $v$, and from $v$ to $i$.
        The dual property also holds, exchanging currents for voltages and series
        interconnection for parallel interconnection.
\end{proposition}

\begin{proof}
        It follows from Lemma~\ref{lem:monotone_properties},
        property~\ref{concatenation}, that the concatenation of $R_1$ and $R_2$ is
        monotone from $(i_1, i_2)$ to $(v_1, v_2)$.  Monotonicity of the
        interconnection from $i$ to $v$ then follows from taking 
        $M = (1, 1)\tran$ in Lemma~\ref{lem:monotone_properties},
        property~\ref{congruence}. Monotonicity from $v$ to $i$ follows by noting that the inverse of a monotone relation
        is monotone (Lemma~\ref{lem:monotone_properties}, property~\ref{inversion}).  
        The proof for a parallel interconnection is identical after exchanging
        currents for voltages and swapping KCL and KVL.
\end{proof}

The negative feedback interconnection of two systems $F$ and $G$ can be represented
as the inverse of the sum of $F^{-1}$ and $G$.  Noting that inversion preserves
monotonicity, the following generalization of the fundamental theorem of passivity
follows from Proposition~\ref{prop:series_parallel}.

\begin{corollary}
        Given two operators $F$ and $G$, each monotone on a scalar Hilbert space
        $\mathcal{H}$, their negative feedback interconnection $(F^{-1} +
        G)^{-1}$ is monotone.
\end{corollary}

\subsection{Strongly monotone and Lipschitz relations}
\label{sec:more_properties}

Here we introduce some additional properties a relation may have.  These properties
are referred to in the next section, as requirements for convergence of some
algorithms.

\begin{definition}
        A relation $S$ has a \emph{Lipschitz constant of} $L>0$ if, for all $(u,
        w), (v, y) \in S$,
        \begin{equation*}
                \norm{u - v} \leq L\norm{w - y}.
        \end{equation*}
        If $L < 1$, $S$ is called a \emph{contraction}.  If $L = 1$, $S$ is called
        \emph{nonexpansive}.
\end{definition}

\begin{definition}
        A relation $S$ is said to be \emph{averaged} if there exists a $\theta
        \in (0, 1)$ such that $S = (1 - \theta)I + \theta G$, where $I$ is the
        identity relation and $G$ is some nonexpansive relation.
\end{definition}

\begin{definition}
        A relation $S$ is $\beta$-\emph{coercive} or \emph{strongly monotone} with
        parameter $\beta > 0$ if, for all $(u, w), (v, y) \in S$,
        \begin{equation*}
                \bra{u -v}\ket{w - y} \geq \beta\norm{u - v}^2.
        \end{equation*}
\end{definition}

\begin{definition}
        A relation $S$ is $\alpha$-\emph{cocoercive} for $\alpha > 0$ 
        if, for all $(u, w), (v, y) \in S$,
        \begin{equation*}
                \bra{u -v}\ket{w - y} \geq \alpha \norm{w - y}^2.
        \end{equation*}
\end{definition}

It is seen immediately that $F$ is $\alpha$-coercive if and only if $F^{-1}$ is
$1/\alpha$-cocoercive.  It also follows from the Cauchy-Schwarz inequality that $F$
has a Lipschitz constant of  $1/\alpha$ if $F$ is $\alpha$-cocoercive.  Finally, if
$A$ is $\alpha$-coercive (resp. $\alpha$-cocoercive) and $B$ is monotone, $A + B$ is
is $\alpha$-coercive (resp. $\alpha$-cocoercive).  For more details on these
properties, we refer the reader to \cite[\S 2.2]{Ryu2020}.

\section{Algorithmic steady-state analysis of monotone systems}
\label{sec:computation}
In this section, we develop an algorithmic method for computing the periodic response
of a maximal monotone system which is forced by a periodic input.

% introduce space of discrete time periodic signals here
A trajectory $w(t)$ is said to be
$T$-periodic if $w(t) = w(t + T)$ for all $t$.  Let $L_{2, T}$ denote the Hilbert
space of finite energy, $T$-periodic, continuous signals, restricted to a single
period, with values in $\R$.  Define the inner product on $L_{2, T}$ by $\bra{x}\ket{y} = \int_0^T
xy \dd{t}$ and the induced norm by $\norm{x} = \sqrt{\bra{x}\ket{x}}$.
Let $l_{2, N}$ denote the discrete-time counterpart of $L_{2, T}$, that is, the space
of discrete, square summable,
real valued signals $x = {x_0, x_1, \ldots}$, where $x_k \in \R$, which are $N$-periodic: $x_0 = x_N$.  
This space may be identified with $\R^N$.

The algorithmic problem we consider is as follows.

\textbf{Problem statement}

% make this for arbitrary systems
%Consider a maximal monotone system represented by a relation $R$.
%Assume that $R$ maps from $u$ to $y$. Suppose that $y$
%is now an arbitrary, zero mean periodic signal, $y^\star$.  Find a corresponding
%periodic $u^\star$ such that $(u^\star, y^\star) \in R$.
Find, if possible, a $T$-periodic $u^\star$ such that $(u^\star, y^\star) \in R$, where
$R$ is a maximal monotone relation on $l_{2, T}$ mapping $u$ to $y$, and $y^\star$ is an
arbitrary periodic signal in the domain of $R$.

\textbf{General solution method}

Define the relation $\Delta R$ by $\Delta R(u) = R(u) - y^\star$.  Then any $u$ which
solves $0 \in \Delta R(u)$ is a possible solution $u^\star$.  The problem $0 \in
\Delta R(u)$ is solved using a fixed point iteration. \\%hfill$\righthalfcup$

The maximal monotone systems we consider are constructed from simple maximal monotone relations
interconnected in ways that preserve maximal monotonicity.  The components may be
defined both by static relations and differential/integral relations.

The derivative is discretized to give the operator $D$, which is the backwards finite
difference, given by the relation
\begin{IEEEeqnarray*}{rCl}
        D = \bigg\{ (u, y) \; \bigg| \; y_T = N D_T u_T, y_{N-1} &=& N \sum_{k =
        0}^{N-2} y_k\\ &=& -N u_{N-2}, u_{N-1} = 0 \bigg\},
\end{IEEEeqnarray*}
where $x_T$ denotes the truncation of the signal $x$ to the first $N-1$ components,
$x_T = \{x_0, x_1, \ldots, x_{N-2} \}$, and $D_T$ is the $N-1 \times N-1$ matrix
\begin{IEEEeqnarray*}{rCl}
D_T &=& 
\begin{bmatrix}
        1 & 0 & 0 & \ldots & 0 & 0\\
        -1 & 1 & 0 & \ldots & 0 & 0\\
        0 & -1 & 1 & \ldots & 0 & 0\\
        \vdots & \vdots & \vdots & \ddots & \vdots &
        \vdots \\
        0 & 0 & 0 & \ldots & -1 & 1
\end{bmatrix}.
\end{IEEEeqnarray*}

The inverse relation $J \coloneqq D^{-1}$, which replaces the integral, is given by
\begin{IEEEeqnarray*}{rCl}
        J &=& \left \{ (u, y) \; \middle| \; y_T = \frac{1}{N} J_T u_T, u_{N-1} = N \sum_{k =
        0}^{N-2} u_k, y_{N-1} = 0 \right\},
\end{IEEEeqnarray*}
where $J_T$ is the $N-1 \times N-1$ Riemann sum:
\begin{IEEEeqnarray*}{rCl}
        J_T  &=& 
\begin{bmatrix}
        1 & 0 & 0 & \ldots & 0 & 0\\
        1 & 1 & 0 & \ldots & 0 & 0\\
        1 & 1 & 1 & \ldots & 0 & 0\\
        \vdots & \vdots & \vdots & \ddots & \vdots &
        \vdots \\
        1 & 1 & 1 & \ldots & 1 & 1
\end{bmatrix}.
\end{IEEEeqnarray*}

In order to have well defined integral signals, the input $u$ to any integral is
restricted to have
zero mean:
\begin{IEEEeqnarray*}{rCl}
       \sum_j u_j &=& 0,
\end{IEEEeqnarray*}
and the output of any integral is restricted to have zero offset: $y_{N-1} = 0$.

 Note that both $D$ and $J$ are maximal monotone relations - this can be checked by noting
 that $D_T$ and $J_T$ are both maximal monotone linear relations (their symmetric parts are
 positive definite) and the additional constraints in each relation do not
 affect maximal monotonicity.  Here we use the backwards finite difference discretization,
 but any discretization of the derivative may be used by defining $D$ and $J$
 appropriately.

\subsection{Fixed point algorithms} 
\label{sec:fixed_point_algorithms}

In this section, we briefly introduce the algorithms which will be used in the
examples that follow.  Rather than give an exhaustive treatment of these methods, we
give only the necessary details, and refer the interested reader to the introductory 
papers \cite{Parikh2013, Ryu2016, Bertsekas2011, Combettes2018}, and the
textbooks \cite{Bauschke2011, Ryu2020}.

A fixed point of a relation $F$ is a point $x$ such that $x \in F(x)$. The fixed
point, or Picard, iteration, is given by the update rule
\begin{IEEEeqnarray*}{rCl}
        x^{j + 1} &=& F(x^j),
\end{IEEEeqnarray*}
 where we use superscripts to denote
iterations of the entire signal $x$, as opposed to subscripts, which denote a
particular time instant of the signal.  For several classes of relation $F$, this
iteration will converge to a fixed point of $F$. This includes contractions
\cite{Banach1922} and averaged relations \cite{Mann1953, Krasnoselskij1955}.

There are a number of ways in which $0 \in S(x)$ can be solved by finding a fixed point of a
relation $F$ related to $S$.  For interconnections of 1-port elements, $S$ corresponds either to a
single element, in which case it is of the form $y = S(u)$, or corresponds to an
interconnection of two smaller 1-ports, in which case it is of the form $y = S_1(u) +
S_2(u)$.  We describe an algorithm corresponding to each of these two forms below.

\subsubsection*{Forward step}
The forward step is suitable for solving a single maximal monotone relation $S$.
It can be viewed as a generalization of gradient descent.
\begin{IEEEeqnarray*}{lrCl}
       & 0 &\in& S(x) \\
        \iff & 0 &\in& -\alpha S(x) + x - x\\
        \iff & x &\in& (I - \alpha S)(x).
\end{IEEEeqnarray*}
The fixed point iteration $x^{j+1} = x^{j} - \alpha S(x^j)$ converges geometrically
to the unique fixed point of $I - \alpha S$ when $S$ is strongly
monotone with parameter $m$ and Lipschitz with parameter $L$, and $\alpha \in (0,
2m/L^2)$.  The optimal contraction factor of $1 - m^2/L^2$ is given by $\alpha = m/L^2$.
Convergence is proved by showing $I - \alpha S$ is a contraction mapping; existence
and uniqueness then follow from the Banach fixed point theorem (see \cite{Ryu2016}).

\subsubsection*{Douglas-Rachford splitting}
When a relation can be split into a sum of two monotone relations, $y = S(u) = S_1(u) +
S_2(u)$, as it can for a series or parallel interconnection of two 1-ports, a \emph{splitting method} 
can be applied to separate the computation for each of the constituent relations, $S_1$
and $S_2$.  This is useful when computation is much simpler for the component relations
than for their sum.  

The Douglas-Rachford splitting algorithm
\cite{Douglas1956, Lions1979} converges if $S_1$ and $S_2$ are maximal monotone
\cite{Svaiter2011}.  The iteration for this algorithm is as follows:
\begin{IEEEeqnarray*}{rCl}
        x^{k + 1/2} &=& \res_{S_1}(i^k)\\
        z^{k + 1/2} &=& 2x^{k + 1/2} - i^k\\
        x^{k+1} &=& \res_{S_2}(z^{k + 1/2})\\
        i^{k+1} &=& i^k + x^{k+1} - x^{k + 1/2}.
\end{IEEEeqnarray*}
$\res_{S}$ denotes the relation $(I + \alpha S)^{-1}$, which is called the
\emph{resolvent} of $S$.
The Douglas-Rachford algorithm forms the basis of the Alternating Direction Method of
Multipliers (see \cite{Boyd2010} and references therein), and is one of the most successful fixed point
methods.

\section{Monotone RLC 1-ports}
\label{sec:1ports}

We illustrate the algorithmic computation of steady-state relations on the important
class of 1-port RLC circuits, formed from the parallel and series interconnection of monotone nonlinear resistors and
LTI capacitors and inductors.  We have already shown in
Proposition~\ref{prop:series_parallel} that these interconnections preserve
monotonicity.

A resistor is a relation $R$ on $\R$ between current and voltage:
\begin{IEEEeqnarray*}{rrCl}
        & R = \left\{(i, v) \in \R \times \R\; |\; v \in R(i)\right\}\\
        \text{or } & R =  \left\{(v, i) \in \R \times \R\; |\; i \in G(v)\right\}.
\end{IEEEeqnarray*}
The first form is the \emph{current controlled} form, the second is the \emph{voltage
controlled} form.  A resistor defines a 1-port relation on $L_{2, T}$ by applying
the relation $R$ at each time:

\begin{IEEEeqnarray*}{rCl}
        S = \left\{(i, v) \in L_{2, T} \times L_{2, T} \; |\; (v(t), i(t)) \in R \text{ for all } t \right\}.
\end{IEEEeqnarray*}

$\phi$ and $v$ are related by Faraday's law of induction:
\begin{IEEEeqnarray*}{rCl}
        \td{}{t} \phi &=& v, \IEEEyesnumber \label{eq:fundamental_pv}
\end{IEEEeqnarray*}

Similarly, $i$ and $q$ are related by 
\begin{IEEEeqnarray*}{rCl}
        \td{}{t} q &=& i. \IEEEyesnumber \label{eq:fundamental_iq}
\end{IEEEeqnarray*}

A capacitor is a relation $C$ on $\R$ between charge and voltage:  
\begin{IEEEeqnarray*}{rrCl}
        & C = \left\{(q, v) \in \R \times \R\; |\; q \in C(v)\right\}\\
        \text{or } & C =  \left\{(v, q) \in \R \times \R\; |\; v \in E(q)\right\}.
\end{IEEEeqnarray*}
An inductor is given by a relation $L$ on $\R$ between flux and current:
\begin{IEEEeqnarray*}{rrCl}
        & L = \left\{(i, \phi) \in \R \times \R\; |\; \phi \in L(i)\right\}\\
        \text{or } & L = \left\{(\phi, i) \in \R \times \R\; |\; i \in A(\phi)\right\}.
\end{IEEEeqnarray*}
Capacitors and inductors define relations on $L_{2, T}$ between charge and voltage
and flux and current, respectively.  Composing these relations with the appropriate fundamental law
\eqref{eq:fundamental_pv} or \eqref{eq:fundamental_iq} gives 1-port relations between current and voltage.
 
The following lemma gives a characterization of the monotonicity of resistors,
capacitors and inductors on $L_{2, T}$ in terms of their devices laws.
%The following lemma shows that the monotonicity
%of resistors and capacitors on $L^2$ is captured by the monotonicity of their 
%device law, defined on $\R$.
\begin{lemma}
        \label{lem:monotone-r-c}
        \begin{enumerate}
                \item A resistor is monotone on $L_{2, T}$ if it defines a monotone relation
                        on $\R$ between $i(t)$ and $v(t)$ for all $t$. \label{part_1}
                \item A capacitor or inductor is monotone on $L_{2, T}$ if
                        its device law is linear, time invariant and has positive
                        gradient. \label{part_2}
        \end{enumerate}
\end{lemma}

\begin{proof}
        We begin by showing that a memoryless relation on $L_{2, T}$ between $i$ and
        $v$
        (that is, a relation where $v(t)$ depends only on $i(t)$) is
        monotone if it is monotone on $\R$.  Indeed, by monotonicity on $\R$, we have
        \begin{IEEEeqnarray*}{rCl}
                (i_1(t) - i_2(t))(v_1(t) - v_2(t)) \geq 0 \text{ for all } t,
        \end{IEEEeqnarray*}
        from which it follows that
        \begin{IEEEeqnarray*}{rCl}
                \bra{i_1 - i_2}\ket{v_1 - v_2} &=& \int_{0}^{T} (i_1(t) -
                i_2(t))(v_1(t) - v_2(t)) \dd{t}\\
                                               &\geq& 0.
        \end{IEEEeqnarray*}
        This proves part~\ref{part_1} of the lemma. 

        We now show that an LTI capacitor is monotone on $L_{2, T}$.  The proof
        follows the lines of \cite[Ex. 20.9]{Bauschke2011}.  The proof for an
        inductor is identical.
        \begin{IEEEeqnarray*}{l}
               \int_0^T (v_1(t) - v_2(t))(i_1(t) - i_2(t))\dd{t} \\ =  \int_0^T
               C(v_1(t) - v_2(t))(\dot{v}_1(t) - \dot{v}_2(t))\dd{t}\\
              = \frac{C}{2} \int^T_0 \td{}{t} |v_1(t) - v_2(t)|^2 \dd{t}\\
              = \frac{C}{2} (|v_1(T) - v_2(T)|^2 - |v_1(0) - v_2(0)|^2)\\
              = 0.
        \end{IEEEeqnarray*}
\end{proof}

The class of systems which can be represented by series and parallel interconnection
of maximal monotone resistors and LTI capacitors and inductors encompasses Lur'e
systems with a passive LTI system in the forward path and a maximal monotone relation
in the return path.  Indeed, if the forward path has a transfer function $G(s)$ and
the feedback path is a relation $R$, the Lur'e system can be synthesized as the
series interconnection of a resistor with resistance relation $R$ and an LTI network
with impedance $G(s)^{-1}$ (which can be synthesized using the Bott-Duffin
construction \cite{Bott1949}).  The $i-v$ relation on $L_{2,
T}$ of this 1-port is $v = (R + H^{-1})^{-1}(i)$, where $H$ is the relation on $L_{2,
T}$ corresponding to $G$.

The periodic response of such Lur'e systems is studied by Heemels \emph{et al.}
\cite{Heemels2017}, who give two algorithms for computing the periodic output of such
a system, given a periodic input.  Given a state space realization, they show that
the system can be represented as a maximal monotone differential inclusion, and that
backwards Euler discretization corresponds to computing the resolvent of this
differential inclusion at each time step.  Their first algorithm involves iteratively 
computing this resolvent forward in time.  Their second algorithm combines the
computation of this resolvent over a period with a periodic boundary condition.
In comparison, the method we present here has several advantages.  Firstly, it is
independent of a state space realization or differential inclusion representation -
the relations used in computation represent the components of the system and retain
their physical meaning.  Secondly, and most importantly, this allows the computation
to be split in a way which corresponds to the structure of the system.  Lastly, our
method applies immediately to any time discretization method (by an appropriate
definition of the $D$ and $J$ relations).

\subsection{Existence of periodic solutions}
We have presented algorithms for finding periodic solutions of periodically driven
systems, without guaranteeing such solutions exist.  Here, we collect a few basic
results which show that such periodic solutions do exist for large classes of
maximal monotone RLC 1-ports.  We begin by noting that for resistors, capacitors, inductors and purely
series or parallel circuits, periodic inputs map to periodic outputs, regardless of
the monotonicity of the elements, provided that signals which are integrated have
zero mean. In the interest of brevity, we omit the proofs of the following two
propositions, which are straightforward.

\begin{proposition}\label{thm:periodic_elements}
        %Resistors map $T$-periodic currents to $T$-periodic voltages, and
        %$T$-periodic voltages to $T$-periodic currents.  Capacitors map $T$-periodic
        %voltages to $T$-periodic currents and zero mean $T$-periodic currents to
        %$T$-periodic voltages.  Inductors map $T$-periodic
        %currents to $T$-periodic voltages and zero mean $T$-periodic voltages to
        %$T$-periodic currents.
        Memoryless relations and the derivative map $T$-periodic inputs to
        $T$-periodic outputs. The integral maps zero mean $T$-periodic inputs to
        $T$-periodic outputs.
\end{proposition}

%\begin{proof}
%        The property is clearly true for a memoryless relation, that is, a relation
%        between $u$ and $y$ such that $y(t) \in f(u(t))$.  Indeed, $y(t + T) \in f(u(t
%        + T)) = f(u(t)) \ni y(t)$. This includes resistors viewed
%        as relations between $i$ and $v$ and capacitors as relations between $q$ and
%        $v$.
%
%        The property also holds for the derivative, and for the integral when the
%        signal has zero mean over a period.  Indeed, 
%        \begin{IEEEeqnarray*}{rCl}
%                \td{u(t)}{t} &=& \lim_{h \rightarrow 0} \frac{u(t) + u(t + h)}{h}\\
%                             &=& \lim_{h \rightarrow 0} \frac{u(t + T) + u(t + T + h)}{h}\\
%                             &=& \td{u(t + T)}{t},
%        \end{IEEEeqnarray*}
%        and
%        \begin{IEEEeqnarray*}{rCl}
%                \int_{-\infty}^t u(t) \dd{t} &=& \int_{-\infty}^{t - t\text{ mod }T} u(t) \dd{t} 
%                + \int_{t - t\text{ mod } T}^t u(t) \dd{t}\\
%                                             &=& \int_{t - t\text{ mod } T}^t u(t) \dd{t}\\
%                                             &=& \int_{t - t\text{ mod } T + T}^{t + T} u(t)
%                                             \dd{t}.
%        \end{IEEEeqnarray*}
%\end{proof}

\begin{proposition}\label{thm:periodic_series_parallel}
        A series (resp. parallel) interconnection of $n$ 1-ports which map $T$-periodic
        currents (voltages) to
        $T$-periodic voltages (currents) also maps $T$-periodic currents (voltages) to $T$-periodic
        voltages (currents).  
\end{proposition}

%\begin{proof}
%        Periodicity is preserved under summation of signals, and therefore
%        preserved by Kirchoff's laws.  Indeed, if $y(t) = u^1(t) + u^2(t)$, and $u^1$
%        and $u^2$ are both $N$-periodic, then $y(t + T) = u^1(t+T) + u^2(t+T) =
%        u^1(t) + u^2(t) = y(t)$.
%\end{proof}

Finally, we show that maximal monotone RLC 1-ports map discrete periodic inputs to
discrete periodic outputs,
provided they contain a resistor which is both coercive and cocoercive and the input
lies within the domain of the system.  This is a
resistor with slope bounded between two strictly positive real numbers. 
Other classes of systems with this property include contractive state space systems
\cite{Sontag2010} and approximately finite memory input/output maps \cite{Sandberg1992}.

\begin{theorem}
        A 1-port, given by a relation on $l_{2, T}$, which is composed of series and
        parallel interconnections of monotone
        resistors and LTI capacitors and inductors, has a unique
        $T$-periodic output for every $T$-periodic input in its domain if
        it contains a resistor which is coercive and cocoercive.
\end{theorem}

\begin{proof}
        The proof has two parts.  Building a 1-port from components involves the
        operations of addition and inversion.  First, we note that the properties of
        coercion 
        and cocoercion are preserved under these operations. Specifically, $G(u) \coloneqq A(u) +
        B(u)$ is coercive and cocoercive if $A$ and $B$ are both monotone and $A$ is
        coercive and cocoercive.  If $G$ is coercive and cocoercive, $G^{-1}$ is also coercive and
        cocoercive.  These two properties show that any circuit that is composed of
        monotone elements with a single element which is coercive and cocoercive is
        coercive and cocoercive both from $i$ to $v$ and from $v$ to $i$.

        The second part of the proof shows that if $G$ is coercive and cocoercive, it
        has a $T$-periodic output for every $T$-periodic input.  We consider the
        1-port as a relation on $l_{2, T}$. $G$ cocoercive with
        parameter $\alpha$ implies $G$ is Lipschtiz with parameter $1/\alpha$.
        Following the computational method we have presented,
        we take, for a given periodic $y^\star$, the incremental relation $\Delta y =
        G(u) - y^\star$.  Convergence of the forwards step algorithm when
        $A$ is strongly monotone and Lipschitz guarantees the uniqueness and existence of a zero to
        the incremental relation.
\end{proof}

Note that the requirement that inputs lie in the domain of the interconnection is
restrictive for some interconnections - for example, if integrators are used (for
instance in capacitors mapping current to voltage), their domain is zero mean
input currents.

\section{Example}
\label{sec:example}

An envelope detector is a simple nonlinear circuit consisting of a diode in series
with an LTI RC filter
 (Figure~\ref{fig:envelope_detector}).  It is used to demodulate AM radio signals.

\begin{figure}[hb]
        \centering
        \begin{circuitikz}[american]
                %\draw (0, 0) to[short, o-, i<_=$i$] (2, 0)
                \draw (0, 0) to[short, o-] (2, 0)
                (0, 2) to[short, o-, i=$i$] (1, 2)
                to[D] (2.5, 2)
                to[C=$C$] (2.5, 0)
                (2, 2) to[short] (4, 2)
                to[R=$R$] (4, 0)
                to[short] (2, 0)
                (0, 2) to[open, v^=$v$] (0, 0);
        \end{circuitikz}
        \caption{An envelope detector, configured as a 1-port.}%
        \label{fig:envelope_detector}
\end{figure}
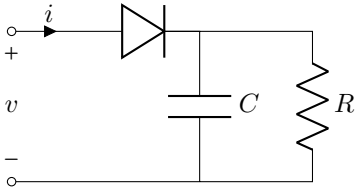

We model the
diode using the Shockley equation:
\begin{equation*}
        v = nV_T \ln\left(\frac{\Delta i + i^\star}{I_s} + 1\right),
\end{equation*}
where $I_s$ is the reverse bias saturation current, $V_T$ is the thermal voltage and $n$
is the ideality factor.

The outermost interconnection in this circuit is the series interconnection of the
diode and filter:
\begin{equation*}
        v = R_{\text{diode}}(i) + R_{RC}(i),
\end{equation*}
where $R_{\text{diode}}$ is the current to voltage relation of the diode, and
$R_{RC}$ is the current to voltage map of the RC filter.

The RC filter is itself a parallel interconnection of a resistor and capacitor, which maps voltage to current:
\begin{IEEEeqnarray*}{rCl}
        R_{RC} &=& G_{RC}^{-1}\\
        i_{RC} &=& G_{RC}(v_{RC}).
\end{IEEEeqnarray*}

\noindent\textbf{Current to voltage simulation}

Computing the voltage given a current input involves applying the relation of
the outer series interconnection.  The voltage response is the sum of the diode
voltage and the filter voltage.  The diode voltage may be computed immediately by
applying the diode relation $R_{\text{diode}}$.  To compute the filter voltage, we
have to invert the filter relation.  This is done using the forward step
algorithm:
\begin{IEEEeqnarray*}{rCl}
        v_{RC}^{j+1} &=& v_{RC}^j - \alpha \Delta G_{RC}(v_{RC}^j),
\end{IEEEeqnarray*}
where $\Delta G_{RC}$ is the forward relation of the RC circuit offset by the
input current:
\begin{IEEEeqnarray*}{rCl}
        \Delta G_{RC}(v_{RC}) = G_{RC}(v_{RC}) - i^\star.
\end{IEEEeqnarray*}
$G$ is the forward relation of the RC circuit, which encompasses the following steps:
\begin{IEEEeqnarray*}{rrCl}
        \text{KVL:}\quad & \begin{bmatrix}
               v_C \\ v_R 
       \end{bmatrix} &=& \begin{bmatrix}
               1 \\ 1
       \end{bmatrix} v\\
        \text{Devices:}\quad & \begin{bmatrix}
               q_C \\ i_R 
       \end{bmatrix} &=& \begin{bmatrix}
               C v_C \\
               \frac{1}{R} v_R
       \end{bmatrix}\\
                \text{Physical law:}\quad & i_C &=& D q_C\\
      \text{KCL:}\quad& i &=& \begin{bmatrix}
      1 & 1 \end{bmatrix}
      \begin{bmatrix}
              i_C \\ i_R
      \end{bmatrix}.
\end{IEEEeqnarray*}
These steps can be collapsed to a single linear operator:
\begin{IEEEeqnarray*}{rCl}
        G = \frac{C}{N}D + \frac{1}{R}I.
\end{IEEEeqnarray*}
As $\Delta G$ is affine, it has a Lipschitz constant $L$ equal to its
largest singular value and is coercive with constant $m = \lambda_{\min}((\Delta G +
\Delta G\tran)/2)$. 
Hence the forward step converges geometrically for any choice of $\alpha \in
(0, 2m/L^2)$.

Figure~\ref{fig:forward_simulation} shows the results of performing this scheme with
an input of $i^\star = 1 + \sin(2\pi t)\si{\ampere}$, with $R = 1\si{\ohm}$, $C =
1\si{\farad}$, $I_s = 1\times10^{-14}\si{\ampere}$, $n = 1$ and $V_T = 0.02585\si{\volt}$.  
The number of time steps used is $500$.

\begin{figure}[h]
        \centering
        \begin{tikzpicture}
                \begin{groupplot}
                        [
                        group style={
                                group size=1 by 2,
                                vertical sep = 0.5cm
                        },
                        width=0.5\textwidth,
                        height=4cm,
                        cycle list name=colors,
                        grid=both,
                        grid style={line width=.1pt, draw=Gray!20}, 
                        axis x line=bottom,
                        axis y line=left
                        ]
                        \nextgroupplot[ylabel={\footnotesize Voltage (V)}, xmin=0,
                        xmax=1]
                        \addplot[CornflowerBlue] table [x = t, y = v, col sep = comma, mark = none]{"./envelope_detector.csv"};
                                \addlegendentry{\footnotesize $v$ - output};
                                
                                \nextgroupplot[xlabel={\footnotesize Time (s)}, ylabel={\footnotesize Current (A)}, xmin=0, xmax=1]
                        \addplot[BurntOrange] table [x = t, y = i, col sep = comma, mark = none]{"./envelope_detector.csv"};
                                \addlegendentry{\footnotesize $i$ - input};
                                
                \end{groupplot}
        \end{tikzpicture}
        \caption{Input current $i^\star$ and the resulting voltage $v$ for an
        envelope detector.  One period of a periodic input and output is shown.}%
        \label{fig:forward_simulation}
\end{figure}
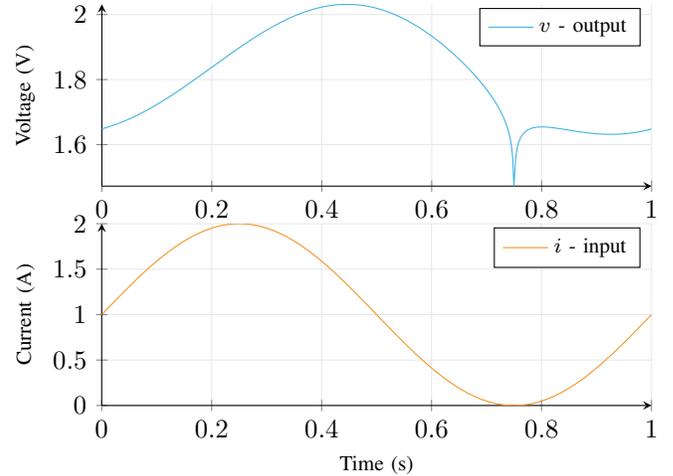

\noindent\textbf{Voltage to current simulation}\\
\indent The voltage to current simulation involves computing the inverse relation of the
outer series interconnection.

The properties of the diode limit the available algorithms that one may apply.  It is
maximal monotone, but it is neither Lipschitz nor strongly monotone, and its range
is not the full space $\R$.  Furthermore, the resolvent of the diode cannot be
computed analytically.  The diode relation is however fully separable in time,
meaning the resolvent values at each time point can be computed in parallel.

The incremental voltage $\Delta v = v - v^\star$ is given as a relation of $i$ by
\begin{IEEEeqnarray*}{rCl}
        \Delta v = R_\text{diode}(i) + R_{RC}(i) - v^\star.
\end{IEEEeqnarray*}

The split nature of the relation into diode and RC circuit parts suggests the use of
a splitting algorithm.  Here we apply the Douglas-Rachford splitting algorithm.  This
involves applying both the resolvents $\res_{RC}$ and $\res_{\text{diode}}$.

The resolvent $\res_{RC}$ is given by $(I + \lambda G_{RC}^{-1})^{-1}$.  This is
computed as follows:
\begin{IEEEeqnarray*}{rCl}
        y &=& \res_{RC}(u)\\
          &=& (I + \lambda G_{RC}^{-1})^{-1} u\\
        (I + \lambda G_{RC}^{-1})y &=& u\\
        (G_{RC} + \lambda I)y &=& G_{RC} u.
\end{IEEEeqnarray*}
This last line is solved using a general purpose linear system solver.

The resolvent of the diode, $\res_{\text{diode}}$, is given by $\res_{\text{diode}}^{-1}(x) = (I + \lambda
R_\text{diode}(x) - \lambda v^\star)$. There
is no analytic expression for this operator.  Rather, the resolvent is computed
numerically using the guarded Newton algorithm \cite{Parikh2013}.  

Figure~\ref{fig:envelope_detector_inverse} shows the results of performing this scheme with
an input of $v^\star = \sin(2\pi t)$ A, with $R = 1\, \Omega$, $C = 1$ F, $I_s =
1\times10^{-14}$ A, $n = 1$ and $V_T = 0.02585$ V.  The number of time steps used is
$500$.

\begin{figure}[h]
        \centering
        \begin{tikzpicture}
                \begin{groupplot}
                        [
                        group style={
                                group size=1 by 2,
                                vertical sep = 0.5cm
                        },
                        width=0.5\textwidth,
                        height=4cm,
                        cycle list name=colors,
                        grid=both,
                        grid style={line width=.1pt, draw=Gray!20},
                        axis x line=bottom,
                        axis y line=left
                        ]
                        \nextgroupplot[ylabel={\footnotesize Voltage (V)}, xmin=0, xmax=1]
                        \addplot[CornflowerBlue] table [x = t, y = v, col sep = comma, mark = none]{"./envelope_detector_inverse.csv"};
                                \addlegendentry{\footnotesize $v$ - input};
                                
                                \nextgroupplot[xlabel={\footnotesize Time (s)}, ylabel={\footnotesize Current (A)}, xmin=0, xmax=1]
                        \addplot[BurntOrange] table [x = t, y = i, col sep = comma, mark = none]{"./envelope_detector_inverse.csv"};
                                \addlegendentry{\footnotesize $i$ - output};
                                
                \end{groupplot}
        \end{tikzpicture}
        \caption{Input voltage $v^\star$ and the resulting current $i$ for an
        envelope detector.  One period of a periodic input and output is shown.}%
        \label{fig:envelope_detector_inverse}
\end{figure}
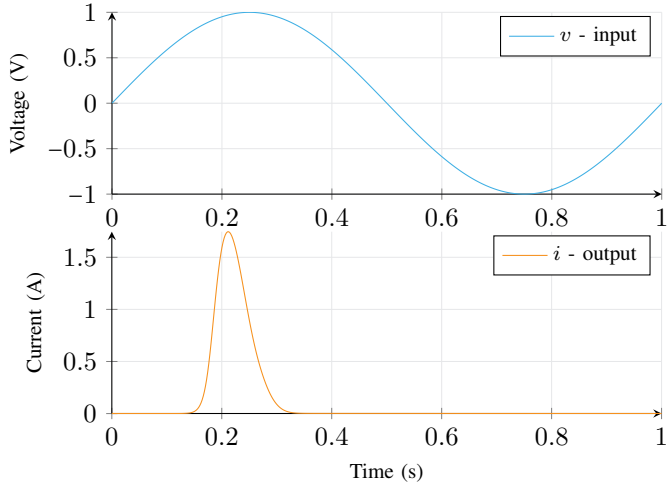

\section{Monotonicity and cyclo-passivity}
\label{sec:cyclo}

Monotonicity is closely related to the classical notions of
cyclo-passivity and passivity.
Cyclo-passivity is a generalisation of 
passivity to storages that are not necessarily bounded, allowing the inference of
instability theorems.  Cyclo-dissipativity was first introduced by
Willems \cite{Willems1974}, and later developed by Hill and Moylan \cite{Hill1975}. 
For recent work on cyclo-dissipativity of multi-ports, see
van der Schaft and Jeltsema \cite{vanderSchaft2020a, vanderSchaft2021} and van der Schaft \cite{vanderSchaft2020}.
To the best of the authors' knowledge, the incremental version of this property has never been
studied.  We define incremental cyclo-passivity as follows:

\begin{definition}
        A relation $R$ on $L_2$ is said to be \emph{incrementally cyclo-passive} if, for all
        $T$, and all $(u, y), (v, w) \in R$ with the property that $u(-T) = u(T)$ and
        $v(-T) = v(T)$ (respectively for $y$ and $w$), then
        \begin{IEEEeqnarray*}{rCl}
                \bra{u - v}\ket{y - w} \geq 0.
        \end{IEEEeqnarray*}
\end{definition}

This corresponds to monotonicity on the space of all periodic trajectories \emph{of any period}.
The property required for the computational methods described in this paper,
monotonicity on periodic signals with a particular period, is a weaker notion.

Monotonicity also relates to the classical input/output notion of incremental
passivity.  Monotonicity is the property that $\bra{u - v}\ket{y - w} \geq 0$ for $y
\in S(u), v \in S(w)$, agnostic of the space that $u$ and $v$ belong to.  If $S$ is
an operator on the extended $L_2$ space (see \cite{Desoer1975}), then 
monotonicity of $S$ is precisely incremental passivity of $S$.  For linear relations,
incremental passivity is equivalent to passivity.

In the linear, time invariant case, the physical property of passivity allows
questions to be answered in a computationally tractable way, for example, a passive
storage function can be found by solving an LMI \cite{Willems1972}.  For
nonlinear passive systems, these computationally tractable methods no longer apply, in
general.  For nonlinear systems with incremental properties, however, tractable
methods do exist.  This is the fundamental result of contraction theory
\cite{Lohmiller1998} and has been noted more recently in dissipativity analysis by
Verhoek, Koelewijn and T\'oth \cite{Verhoek2020} and Forni, Sepulchre and van der
Schaft \cite{Forni2013c}.  The approaches in these works differ from that
of this paper, however, in their reliance on differentiable state space models and 
state-dependent linear matrix inequalities, rather than monotone operator methods.

\section{Conclusions}
\label{sec:conclusions}

We have applied monotone operator optimization methods to the problem of computing
the periodic output of a periodically forced, maximal monotone system.  Splitting
methods allow the algorithm to be separated in a way which mirrors the structure of
the system.  This method has been demonstrated on the class of electrical 1-port
circuits formed from series and parallel interconnections of maximal monotone
resistors and LTI capacitors and inductors.

Early work by Desoer and Wu \cite{Desoer1974} used maximal monotonicity to prove
existence and uniqueness of solutions to networks where not only the resistors, but
also the capacitors and inductors, are described by maximal monotone relations.  The
question of whether the method developed in this paper applies to this broader class
of systems is left open for future research. The treatment of supplies
other than $i\tran v$ is also left open for future research.  
A compelling open question is that of monotone 1-port
synthesis - given an arbitrary monotone relation between current and voltage, when
can it be synthesised as the port behavior of a 1-port built from basic monotone
elements?

%\printbibliography

\bibliographystyle{IEEEtran}
\bibliography{IEEEabrv,monotone} % no whitespace allowed in this command!

\end{document}